\documentclass{amsart}

\usepackage{amsrefs,upgreek,enumerate}

\usepackage{amsthm}
\newtheorem{theorem}{Theorem}[section]
\newtheorem{corollary}[theorem]{Corollary}
\newtheorem{lemma}[theorem]{Lemma}

\newtheorem{remark}{Remark}[section]

\newcommand{\D}{\mathrm{D}}
\newcommand{\e}{\mathrm{e}}
\newcommand{\ii}{\mathrm{i}}
\newcommand{\dd}{\mathrm{d}}
\newcommand{\upi}{\uppi}
\newcommand{\re}{\mathfrak{R}}
\newcommand{\im}{\mathfrak{I}}
\newcommand{\loc}{\mathrm{loc}}
\newcommand{\OO}{\mathrm{O}}

\author{Andrzej Hanyga}
\address{ul. Bitwy Warszawskiej 1920r. 14/52\\
Warszawa, PL\\
{\tt ajhbergen@yahoo.com}}

\title{Asymptotic estimates of viscoelastic Green's functions near the wavefront}

\begin{document}
\maketitle

\begin{abstract}
Asymptotic behavior of viscoelastic Green's functions near the wavefront is expressed in terms 
of a causal function $g(t)$ defined in \cite{SerHanJMP} in connection with the Kramers-Kronig 
dispersion relations. Viscoelastic Green's functions exhibit a discontinuity at the wavefront 
if $g(0) < \infty$. Estimates of continuous and discontinuous viscoelastic Green's functions 
near the wavefront are obtained. 
\end{abstract}

\textbf{Keywords.} viscoelasticity, wavefront, attenuation, dispersion, shock wave, Bernstein function

\section*{Notation.}

{\small
\begin{tabular}{lll}
$\D^n f(t)$ & $\dd^n \, f(t)/\dd t^n$\\
$f\ast g $ & Volterra convolution & $\int_0^t f(s)\, g(t-s) \, \dd s$\\
$\mathcal{L}(f) = \tilde{f}$ & Laplace transform of $f$ & $\int_0^\infty \e^{-p\, t}\, f(t)\, \dd t$\\
$]a,b]$ & the set $a < x \leq b$ &\\
$\re z, \im z$ & real and imaginary part of $z$ &\\
$J(t), J^\prime(t)$ & creep compliance, creep rate 
\end{tabular}
}

\section{Introduction.}
 
It is fair to say that seismology and particularly seismic inversion theory is implicitly 
based on the assumptions that the medium supports discontinuity waves. This assumption 
ensures that seismic signals travel with the wavefront speed while retaining their shape 
except for scaling. However many published explicit models of viscoelastic media such 
as the fractional Zener models 
\cite{CaputoMainardi,BagleyTorvik3,NasholmHolm2013,MainardiVE} used in among others in seismology 
as well as the Cole-Davidson and Havriliak-Negami models \cite{HanWM2013,HavriliakNegami66}  
used for modelling mechanical response of polymers all lack this property. All these models 
are characterized by asymptotic power law behavior of the attenuation function (i.e. the logarithmic 
attenuation rate expressed as a function of circular frequency) in the high frequency range,
which entails continuity of the viscoelastic Green's function 
and all its temporal and spatial derivatives at the wavefront. A viscoelastic pulse in such a medium 
is preceded by a pedestal, i.e. a flat precursor that follows immediately after the wavefront
\cite{Strick1:ConstQ,HanQAM} and therefore it travels at a slower effective speed depending on the propagation time.
(We ignore here the seismological and acoustic models in which signals propagate at infinite speeds
\cite{Kjartansson:ConstQ,CaCaMaHa,StrakaMeerschaertMcGough2013}.). A seismic signal travels with the wavefront if the wavefront of Green's function supports a discontinuity of the stress--velocity
field and lags behind the wavefront if Green's function exhibits the pedestal effect 
\cite{HanQAM}. Relevance of the pedestal effect for seismic inversion was demonstrated in 
\cite{FastTrack} and for borehole velocity surveys in \cite{Strick3:ConstQ}. 

Seismological literature, however, knows two exceptional models which are consistent with the assumptions 
made in seismology: the Jeffreys-Lomnitz creep and a creep model due to Strick and Mainardi. 
In the Jeffreys-Lomnitz and Strick-Mainardi viscoelastic media attenuation is bounded in the 
high frequency range and discontinuities are not immediately smoothed out. 
The last-mentioned models are considered in some detail in \cite{HanDue}.  
The results obtained in this paper apply to both kinds of viscoelastic media. Besides 
we establish here a new criterion for existence of discontinuity waves.

In \cite{SerHanJMP} and \cite{HanWM2013} a theory of attenuation and dispersion in 
general viscoelastic media 
with completely monotone relaxation moduli (or, equivalently, with creep compliances 
that are Bernstein functions) 
was developed. An application of this theory to the study of the relation between
the high-frequency asymptotics of the attenuation function and the wavefront singularities of 
Green's function of an acoustic equation was made in \cite{HanJCA}. These results 
carry over to the viscoelastic equations of motion in one-dimensional space without 
any modification. 
It was shown in \cite{HanJCA} that Green's functions with an unbounded attenuation function 
do not exhibit wavefront discontinuities, while those with an attenuation function
increasing faster than the logarithmic rate in the high-frequency range exhibit the 
pedestal effect. If the attenuation function grows at the logarithmic rate in the 
high-frequency range then the regularity of the wavefront gradually increases with time. 

Related results have been obtained in earlier papers.
Lokshin \cite{LokshinSuvorova82} proved that viscoelastic Green's functions are infinitely differentiable at the
wavefront if $J^\prime(t)/\vert \ln(t)\vert$ is bounded or tends to infinity as $t \rightarrow 0$. Desch and Grimmer obtained
the same result if the relaxation kernel has a logarithmic or stronger singularity \cite{DeschGrimmer85}. The same authors showed that 
in the gap between boundedness and logarithmic singularity more complex regularity patterns are observed, such as 
stepwise regularization \cite{DeschGrimmer89a,HrusaRenardy85,HanJCA}, which we shall also demonstrate here. 

In this paper we shall use a causal function underlying the Kramers-Kronig dispersion relations
\cite{SerHanJMP,HanWM2013} to study the behavior of viscoelastic 
Green's functions near the wavefront in more detail. In particular we shall derive a 
well-known criterion for existence of shock waves in 
viscoelastic media.  This criterion 
probably also applies to nonlinear viscoelastic media under some assumptions on the nature of non-linearity \cite{Greenberg68,ReHrNo:VE}.  The new approach gives more detailed information
about the wavefront asymptotics of Green's functions. For example, for
discontinuous Green's functions it provides the information about the magnitude of the 
wavefront jump discontinuity as well as on the behavior of Green's function immediately 
behind the wavefront.

In \cite{SerHanJMP,HanWM2013} a causal function $g$ of physical dimension 1/L was introduced in such a way that the 
attenuation function $\mathcal{A}(\omega)$, representing the logarithmic attenuation rate 
expressed as a function of circular frequency,
is equal to the real part of $(-\ii \omega)\,\tilde{g}(-\ii \omega)$ 
while the dispersion function $\mathcal{D}(\omega) := \omega\,(1/c(\omega) - 1/c_0)$ is given by the
the expression $\im [\ii \omega \, \tilde{g}(-\ii \omega)]$. 
$c(\omega)$ denotes here the phase speed of a wavefield oscillating with the circular frequency $\omega$.
Existence of the function $g$ is thus linked to the validity of the Kramers-Kronig dispersion relations. 

In this paper we shall show that 
the function $-g(t - \vert x \vert/c_0) \, \vert x \vert$ is an asymptotic phase function near the wavefront of viscoelastic 
Green's functions, where $c_0 := \sup_\omega c(\omega)$ is the wavefront speed and it is assumed that $c_0 < \infty$. 
The theory is based on the 
assumption that the relaxation modulus is locally integrable and completely monotonic (LICM). This assumption
is equivalent \cite{Molinari,HanDuality} to the assumption that the creep compliance $J(t)$ is a Bernstein function 
\cite{BernsteinFunctions}. This assumption is crucial for the existence of the function $g$. 
The assumption that the wavefront speed $c_0$ is finite is equivalent to the inequality $J_0 := 
J(0) > 0$. If $\lim_{t\rightarrow 0} g(t) < \infty$ then Green's function has a jump discontinuity at $\vert x \vert = c_0\, t$. 

It will also be shown that under some additional assumptions valid for many explicit models viscoelastic Green's function
is bounded by $C\, \exp(-g(t-\vert x\vert/c_0)\, \vert x\vert)$ for $t > \vert x\vert/c_0$, where 
$C$ is some constant. 

These results provide an estimate of Green's function near the wavefront and answer the question whether the medium
supports propagation of discontinuities at the wavefront. 
The methods developed in this paper and in \cite{HanWM2013} are applied in \cite{HanDue}
to two examples of viscoelastic media which support discontinuity waves. 

\section{Mathematical preliminaries.}

We shall consider the Initial-Value Problem (IVP) 
\begin{gather} \label{eq:mDlinpro}
\rho \, u_{,tt} = \nabla\cdot[G(t)\ast \nabla u_{,t}] + \delta(x)\,\delta(t), 
\qquad t \geq 0,\quad x \in \mathbb{R}\\
u(0,x) = 0; \quad u_{,t}(0,x) = 0 \label{eq:IVmDlinpro}
\end{gather} 
for the particle velocity $u$ in a hereditary viscoelastic medium. 
Duhamel's principle holds for eq.~\eqref{eq:mDlinpro}, hence the above IVP is equivalent to the IVP 
\begin{gather} \label{eq:mDlinpro1}
\rho \, u_{,tt} = \nabla\cdot[G(t)\ast \nabla u_{,t}], 
\qquad t \geq 0,\quad x \in \mathbb{R}\\
u(0,x) = 0; \quad u_{,t}(0,x) = \delta(x)/\rho \label{eq:IVmDlinpro1}
\end{gather} 
Since we are going to consider Green's functions with jump discontinuities the field $u(t,x)$ will be
considered as the velocity field. 

It is assumed that 
the relaxation modulus $G(t)$ (defined for $t > 0$) is completely monotonic, i.e. it has derivatives $\D^n\, G$ of
arbitrary order and these derivatives satisfy the inequalities 
$$(-1)^n \, \D^n\, G(t) \geq 0 \qquad\text{on $\mathbb{R}$ for $n =0,1,2\ldots$}$$
It is also assumed that $G$ is locally integrable, or, equivalently
$$\int_0^1 G(s)\, \dd s < \infty,$$
In short, $G$ is LICM. It follows \cite{HanDuality} that the creep compliance $J(t)$ 
($t \geq 0$), related to the relaxation modulus $G(t)$ by the equation
\begin{equation} \label{eq:duality}
\int_0^t G(s)\, J(t - s)\, \dd s = t \qquad\text{for $t \geq 0$}
\end{equation}
is a Bernstein function (BF), i.e. it is non-negative, differentiable and its derivative
$J^\prime$ is LICM. Conversely, for a given BF $J$ eq.~\eqref{eq:duality} has a unique 
solution $G$ and $G$ is a LICM function \cite{HanDuality}.
We also recall that $J(t)$ tends to a finite limit $J_0$ as $t \rightarrow 0$ 
$0 \leq J_0 < \infty$ and $J_0 = 0$ if and only if 
$\lim_{t\rightarrow 0+} G(t) = \infty$.

The solution of the IVP (\ref{eq:mDlinpro}--\ref{eq:IVmDlinpro}) is given by the formula
\begin{equation} \label{eq:Green1D}
u(t,x) = \frac{1}{4 \upi \ii} \int_{-\ii \infty + \varepsilon}^{\ii \infty + \varepsilon}
\frac{\kappa(p)}{2 \rho\, p^2}\, \e^{p\, t -\kappa(p)\, \vert x \vert} \, \dd p
\end{equation}
where \begin{equation} \label{eq:kappadef}
\kappa(p) := \rho^{1/2}\, p \, \left[p \, \tilde{J}(p)\right]^{1/2}
\end{equation}
and $\varepsilon > 0$.

In \cite{SerHanJMP,HanWM2013} it was shown that $\kappa(p)$ is a complete Bernstein function 
\cite{BernsteinFunctions,Jacob01I}, i.e.
\begin{equation}\label{eq:CBF}
\kappa(p) = p^2\, \tilde{F}(p),
\end{equation}
 where $F$ is a Bernstein function and $\kappa(0) = 0$.
Consequently $\kappa$ has an integral representation of the following form \cite{HanWM2013} 
\begin{equation} \label{eq:kappa}
\kappa(p) = \frac{p}{c_0} + p \int_{]0,\infty[} \frac{\nu(\dd r)}{p + r}
\end{equation}
where $\nu$ is a positive Radon measure satisfying the inequality
\begin{equation} \label{eq:doss}
\int_{]0,\infty[} \frac{\nu(\dd r)}{1 + r} < \infty 
\end{equation}
and $c_0$ is a constant satisfying the inequalities $0 < c_0 \leq \infty$, defined by
the formula 
\begin{equation}
1/c_0 := \rho^{1/2}\,\lim_{p\rightarrow \infty} \left[p\, \tilde{J}(p)\right]^{1/2} = 
[\rho\, J_0]^{1/2}
\end{equation}
The physical dimension of $\kappa(p)$ and $\nu(\dd r)$ is 1/L. 
 
If $J_0 > 0$ then the constant $c_0$ is finite and it defines the wavefronts $\vert x\vert = c_0\, t$ such that 
$u(t,x) = 0$ for $t > \vert x \vert/c_0$, otherwise $c_0 = \infty$ and the solution 
$u(t,x)$ does not vanish anywhere in the space-time. The wavefront speed $c_0$ has the dimension
L/T \cite{HanWM2013}. In this paper we shall be interested in the case of $c_0 < \infty$.

Recall that every LICM function $\varphi$ has the integral representation 
\begin{equation} \label{eq:LICM}
\varphi(t) = a + \int_{]0,\infty[} \e^{-r\, t} \,\nu(\dd r)
\end{equation}
where $\nu$ is a positive Radon measure satisfying the inequality \eqref{eq:doss}
\cite{GripenbergLondenStaffans}.
Define the function $g$ by the formula
\begin{equation} \label{eq:KK2}
g(t) = \int_{]0,\infty[} \e^{-r\, t} \, \nu(\dd r)
\end{equation}
where the Radon measure $\nu$ is defined by eq.~\eqref{eq:kappa}. 
We then have an important formula
\begin{equation}\label{eq:important}
\kappa(p) = p/c_0 + p\, \tilde{g}(p)
\end{equation}
By Bernstein's Theorem and \eqref{eq:doss} the function 
$g$ is LICM and $\lim_{t\rightarrow \infty} g(t) = 0$. The dimension of $g(t)$ is 1/L. 
The function $g(t)$ assumes a finite value at 0 if $\nu$ has a finite mass. Note that 
any function $\kappa$ given by eq.~\eqref{eq:important}, where $g$ is a LICM function, is a complete Bernstein function.
Indeed, equation~\eqref{eq:CBF} holds for the Bernstein function $F(t) = 1/c_0 + \int_0^t g(s)\, \dd s$.

The Radon measure $\nu$ can be calculated using equation~\eqref{eq:kappa}. If $\nu(\dd r) = h(r)\, \dd r$, 
then
\begin{equation} \label{eq:hrkappa}
h(r) = \frac{1}{\upi} \im \left[ \kappa(p)/p\right]_{p=r \, \exp(-\ii \upi)}
\end{equation}
\cite{SerHanJMP,HanWM2013}, or, using equation~\eqref{eq:kappadef},
\begin{equation} \label{eq:hrJ}
h(r) = \frac{\rho^{1/2}}{\upi} \im \left\{ \left[p \, \tilde{J}(p)\right]^{1/2}\right\}
\end{equation}

Note that $\kappa(p) = p/c_0 + p^2 \, \tilde{F}(p)$,
where $F(t) = \int_0^t g(s)\,\dd s$ ($t \geq 0$). 
Since $g \in \mathcal{L}^1_{\mathrm{loc}}([0,\infty[)$, the function $F(t)$ is continuous on 
$[0,\infty[$ and $F(0) = 0$. The function $F(t)$ can therefore be extended to a continuous 
causal function $\Phi(t)$ on $\mathbb{R}$ and therefore it is a causal distribution. The distributional derivative 
$\phi$ of $\Phi$ is a causal distribution which coincides with the function $g(t)$ on $]0,\infty[$. The analytic 
continuation of $\kappa$ 
to the imaginary axis is the Fourier transform of a causal distribution $D^2\,\Psi(t)$
where $\Psi(t) := \theta(t)/c_0 + \Phi(t)$, hence the real part $\mathcal{A}(\omega) := \re \kappa(-\ii \omega)$ 
of $\kappa(-\ii \omega)$ (the attenuation function) 
and its imaginary part $\im \kappa(-\ii \omega) =: -\mathcal{D}(\omega)$ ($\mathcal{D}$ is the dispersion function) are related
by the Kramers-Kronig dispersion relations \cite{SerHanJMP}. 

In this paper the function $g(t)$ will be used to obtain asymptotic estimates 
and bounds for the solution $u(t,x)$ of the IVP (\ref{eq:mDlinpro}--\ref{eq:IVmDlinpro}).
It will also be seen that $u(t,x)$ has a discontinuity at the wavefront if and only if
$g(0+) := \lim_{t\rightarrow 0+} g(t)$ is finite. The function $g$ is bounded
at 0 if the attenuation-dispersion measure $\nu$ has finite mass; in this case
$g(0+)$ is equal to the total mass of $\nu$. This condition will turn out to be also equivalent 
to the inequality $J^\prime(0) < \infty$. The amplitude of the discontinuity 
will be expressed in terms of $g(0+)$ and also by $\lim_{t\rightarrow 0+} J^\prime(t)$. 

The function $g(t)$ will also be applied to derive an upper bound for the solution near 
the wavefront when $g(0+) = \infty$ and the solution $u(t,x)$ is continuous at the wavefront
for $t > 0$.

\section{Asymptotic behavior of Green's function near the wavefront.}

In this section an asymptotic estimate of Green's function near the wavefront will be derived.

\begin{theorem} \label{thm:BernsteinNew}
The function $x \, f(x)$ on $\mathbb{R}_+$ is CM if and only if there is a non-negative
right-continuous non-decreasing real function $g$ vanishing on $]-\infty,0[$, such that $f = \tilde{g}$.
\end{theorem}
\begin{proof}
If $f = \tilde{g}$ then 
$$ x\, f(x) = -\int_0^\infty \frac{\partial \e^{-x y}}{\partial y} g(y) \, \dd y =
g(0+) + \int_0^\infty \e^{-x y} \, \dd g(y).$$
Define the Radon measure $\mu$ by the formula $\mu(]0,y]) = g(y) - g(0+)$. The measure 
$\mu$ is positive,
hence the Bernstein Theorem implies that $x \, f(x)$ is CM. 

If $x \, f(x)$ is CM then, by the Bernstein Theorem, there is a positive Radon measure 
$\mu$ on $\overline{\mathbb{R}_+}$ such that
$$x\, f(x) = \mu(\{0\}) + \int_{]0,\infty[} \e^{- x y} \, \mu(\dd y)$$
If $g(y) = 0$ for $y < 0$, $g(0+) = \mu(\{0\})$ and $g(y) = \mu(]0,y])$ for $y > 0$
then $g(y) \geq 0$, $g$ is non-decreasing and 
$x \, f(x) = g(0+) + \int_{]0,\infty[} \e^{-x y}\, \dd g(y) = 
x \int_0^\infty \e^{-x y}\, g(y)\, \dd y$.

\end{proof}

\begin{lemma} \label{lem:exp}
If $(-1)^n \, f^{(n)}(x) \geq 0$ for $n = 1,2,\ldots$, then 
$$ (-1)^n \frac{\dd^n}{\dd x^n} \e^{f(x)} \geq 0 \qquad\text{for n =0,1,2,\ldots} $$
\end{lemma}
\begin{proof}
The thesis is obviously true for $n = 0$ and for $n = 1$:
$\dd \,\e^{f(x)}/\dd x = f^\prime(x) \, \e^{f(x)} \leq 0$.

We now assume that for a fixed $n \geq 1$ the identity 
$\dd^n \e^{f(x)}/\dd x^n = P_n(x)\, \e^{f(x)}$ holds, 
where $P_n(x)$ is a polynomial in the derivatives of $f$ with positive coefficients 
such that the sum of the
orders of all the derivatives in a monomial of $P_n$ is $n$. The assumption is
certainly true for $n = 1$. From this assumption it follows
that $\dd^{n+1} \e^{f(x)}/\dd x^{n+1} = P_{n+1}(x)\, \e^{f(x)}$, where 
$P_{n+1}(x) := P_n^\prime(x) + f^\prime(x)\,P_n(x)$. The function $P_{n+1}$ is 
a sum of products of derivatives of $f$ whose orders sum up to $n+1$ with positive coefficients. We have thus 
proved that our assumption is true for all integer $n \geq 1$. But the hypothesis of the lemma implies that
$(-1)^n \, P_n(x) \geq 0$, which proves the lemma.

\end{proof}

\begin{lemma} \label{lem:gpg1}
If $\varphi: \mathbb{R}_+ \rightarrow \mathbb{R}$ is differentiable, non-negative and non-increasing, 
$\lim_{t\rightarrow 0+} [t\, \varphi(t)] \rightarrow 0$ and the Laplace transform $\tilde{\varphi}(p)$ 
of $\varphi$ exists for $p > 0$, then $\exp\left( -p\, \tilde{\varphi}(p)\right)/p$ is
the Laplace function of a non-negative non-decreasing function 
$H: \mathbb{R}_+ \rightarrow \mathbb{R}$.
\end{lemma}
\begin{proof}
We shall first prove that $(-1)^n \, \dd^n \left[p\, \tilde{\varphi}(p)\right]/\dd p^n \leq 0$
for $n =1,2,\ldots$.

We begin with $n = 1$. Let $\varepsilon > 0$.
\begin{multline*}
\frac{\dd}{\dd p} \left[ p \int_\varepsilon^\infty \varphi(t)\,  \e^{-p t}\, \dd t\right] = 
-\frac{\dd}{\dd p} \left[ \int_\varepsilon^\infty \varphi(t)\, \dd \e^{-p t}\right]  =\\
\frac{\dd}{\dd p} \int_\varepsilon^\infty \varphi^\prime(t)\, \e^{-p t}\, \dd t - 
\frac{\dd}{\dd p} \left[\varphi(t)\, \e^{-p t}\right]^\infty_{t=\varepsilon} = 
-\int_\varepsilon t\, \varphi^\prime(t) \, \e^{-p t} \, \dd t + \varepsilon\, \varphi(\varepsilon)
\end{multline*}
%Since $\varphi \in \mathcal{L}^1_\loc$ is non-increasing and non-negative 
%$$\varepsilon \, \varphi(\varepsilon) \leq \int_0^\varepsilon \varphi(t) \dd t \rightarrow 0$$
%as $\varepsilon \rightarrow 0$.
Hence
\begin{equation}
\frac{\dd}{\dd p} \left[p \, \tilde{\varphi}(p) \right] = - \int_0^\infty t \, \varphi^\prime(t)\,
\e^{-p t} \, \dd t
\end{equation}

In particular we have also proved that $t\, \varphi^\prime(t)$ is locally integrable
on $[0,\infty[$. 

For arbitrary $n \geq 1$ we note that
$$\frac{\dd^n}{\dd p^n} \left[ p\, \tilde{\varphi}(p)\right] = -\frac{\dd^{n-1}}{\dd p^{n-1}} 
\int_0^\infty t\, \varphi^\prime(t) \, \e^{-p t} \, \dd t = (-1)^n \int_0^\infty t^n \, \varphi^\prime(t) \, \e^{-p t} \, \dd t$$
But $\varphi^\prime(t) \leq 0$ almost everywhere, hence 
\begin{equation}
(-1)^n \, \frac{\dd^n}{\dd p^n} \left[ p\, \tilde{\varphi}(p)\right] \leq 0, \qquad n = 1,2,\ldots
\end{equation}

Lemma~\ref{lem:exp} implies that 
$$(-1)^n \, \frac{\dd^n}{\dd p^n} \e^{-p\, \tilde{\varphi}(p)} \geq 0\qquad\text{for $n= 0,1,2,\ldots$}$$
Theorem~\ref{thm:BernsteinNew} now implies that 
\begin{equation}
\e^{-p\, \tilde{\varphi}(p)}/p = \int_0^\infty \e^{-p t} \, H(t)\, \dd t
\end{equation}
where the function $H$ is non-negative and non-decreasing.

\end{proof}

\begin{corollary} \label{cor:zcv}
If $\varphi(0+) = \lim_{t\rightarrow 0+} \varphi(t)$ exists then 
$$\lim_{p\rightarrow \infty} \e^{-p\, \tilde{\varphi}(p)} = \exp\left(-\lim_{p\rightarrow \infty}
\left[p\, \tilde{\varphi}(p)\right]\right) = \e^{-\varphi(0+)}$$ and
the limit $H(0+)$ exists and equals $\e^{-\varphi(0+)}$.

Similarly, if the limit $\varphi_\infty = \lim_{t\rightarrow \infty} \varphi(t)$ exists, then 
$\lim_{t\rightarrow \infty} H(t) = \e^{-\varphi_\infty}$.
\end{corollary}

\begin{lemma} \label{lem:xcv}
If the function $\varphi \in \mathcal{L}^1_\loc(\mathbb{R}_+)$ is non-negative, non-increasing 
and convex, $\lim_{t\rightarrow 0+} \left[t\,\varphi^\prime(t)\right] = 0$
then
$\lim_{p\rightarrow \infty} \left[p \, \tilde{\varphi}(p) - \varphi(1/p)\right] = 0$.
\end{lemma}
\begin{proof}
The Laplace transform $\tilde{\varphi}$ of $\varphi$ is defined and
\begin{multline*}
p \, \tilde{\varphi}(p) - \varphi(1/p) = p \int_0^\infty \varphi(t)\, \e^{-p t} \, \dd t- 
\varphi(1/p) = \\ \int_0^\infty \left[\varphi(t/p) - \varphi(1/p)\right]\, \e^{-t} \, \dd t = 
I_1 + I_2
\end{multline*}

We begin with an estimate of 
$I_1 = \int_0^1 \left[\varphi(t/p) - \varphi(1/p)\right]\, \e^{-t} \, \dd t$. The function 
$\varphi$ is non-decreasing and $t \leq 1$, the integrand is non-negative, hence $I_1 \geq 0$.
On the other hand 
\begin{multline*}
I_1 \leq \int_0^1 \left[\varphi(t/p) - \varphi(1/p) \right] \, \dd t =
p \int_0^{1/p} \varphi(q) \, \dd q -\varphi(1/p) = -p \int_0^{1/p} q\, \varphi^\prime(q)  \,
 \dd q  = \\ -p \int_0^{1/p} q \, \varphi^\prime(q) \, \dd q + \left[p \, q \, \varphi(q)\right]_{q=0}^{q=1/p} - \varphi(1/p) = -p \int_0^{1/p} q\, \varphi^\prime(q) \, \dd q 
\end{multline*}
By Rolle's Theorem there is a $q_1 \in [0,1/p]$ such that 
$$I_1 \leq -q_1\, \varphi^\prime(q_1)$$
If $p \rightarrow \infty$, then $q_1 \rightarrow 0$ and, by assumption, $q_1 \, \varphi^\prime(q_1) \rightarrow 0$. Consequently $I_1 \rightarrow 0$ as $p \rightarrow \infty$.

We now turn to $I_2 := \int_1^\infty \left[\varphi(t/p) - \varphi(1/p) \right] \, \dd t$.
For $t > 1$ we have $0 \leq \varphi(1/p) - \varphi(t/p) = 
\int_{1/p}^{t/p} \vert \varphi^\prime(q)\vert \, \dd q$. Convexity of $\varphi$ implies that $\vert \varphi^\prime(q)\vert$ is non-increasing, hence      
$0 \leq \varphi(1/p) - \varphi(t/p) \leq (t/p) \, \vert \varphi^\prime(1/p)\vert$.
Hence $$0 \leq -I_2 \leq \frac{1}{p} \vert \varphi^\prime(1/p)\vert  \int_1^\infty t \, \e^{-t} \, \dd t.$$
On account of our assumptions the first factor on the right-hand side tends to 0 as $p \rightarrow \infty$. 
Hence $I_2 \rightarrow 0$ as $p \rightarrow \infty$, which proves the thesis.

\end{proof}

\begin{lemma}\label{lem:ycv}
Let $0 < \lambda \leq 1$. If $\varphi$ is non-negative, differentiable, non-increasing and convex and 
$t \varphi^\prime(t) \rightarrow 0$ for $t \rightarrow 0$,
then  $0 \leq \varphi(\lambda/p) - \varphi(1/p) \rightarrow 0$ as $p \rightarrow \infty$.
\end{lemma}
\begin{proof}
\begin{multline*}
0 \leq \varphi(\lambda/p) - \varphi(1/p) = \int_{\lambda/p}^{1/p} \vert \varphi^\prime(q)\vert \, \dd q 
\leq \frac{1-\lambda}{p} \vert \varphi^\prime(\lambda/p)\vert \leq \\
\frac{1}{p} \vert \varphi^\prime(\lambda/p)\vert = 
\frac{1}{\lambda} \frac\lambda{p} \vert \varphi^\prime(\lambda/p)\vert \rightarrow 0
\end{multline*}
in view of the last hypothesis.
 
\end{proof} 

We are now ready to investigate the asymptotic behavior of the function $H(t,r)$ at $t = 0$.
\begin{theorem} \label{thm:asymp}
\begin{equation} \label{eq:asymp}
H(t,r) \sim_{t\rightarrow 0} \e^{-g(t)\, r}, \qquad t > 0
\end{equation} 
\end{theorem}
\begin{proof}
The Laplace transform of $H(t,r)$ is $\exp(-p\, \tilde{\varphi}(p))/p$, where 
$\varphi(t) = g(t)\, r$. We now consider the limit at $t \rightarrow 0$ of the
ratio
$$R := \exp\left(-(p/\lambda) \, \tilde{\varphi}(p/\lambda)  \right)/
\exp\left(-p \; \tilde{\varphi}(p)\right),$$
where $\lambda > 0$ is arbitrary. Lemma~\ref{lem:ycv} implies that the limit of R at $p = \infty$ will
not change if we divide it by $\e^{-\varphi(\lambda/p) + \varphi(1/p)}$.
Hence 
$$\lim_{p\rightarrow\infty} R = \lim_{p\rightarrow \infty} \exp\left(-p/\lambda\;\tilde{\varphi}(p/\lambda) + \varphi(\lambda/p)\right) \times \lim_{p\rightarrow\infty} \exp\left(p\,
\tilde{\varphi}(p) - \varphi(1/p) \right)$$
Lemma~\ref{lem:xcv} implies that both limits on the right-hand side are equal to 1,
hence $\lim_{p\rightarrow \infty} R = 1$ for all $\lambda > 0$. Consequently the function 
$l(p) := \exp\left(-p \; \tilde{\varphi}(p)\right)$ is slowly varying at infinity.

But $\tilde{H}(p,r) = l(p)/p$ and by the Karamata Tauberian theorem (Appendix)
$H(t,r) \sim_0 l(1/t)  \equiv \e^{-\tilde{\varphi}(1/t)/t}$. By Lemma~\ref{lem:xcv}
$H(t,r) \sim_0 \e^{-\varphi(t)} \equiv \e^{-g(t)\, r}$.

\end{proof} 

It remains to link the estimates of $H(t,r)$ to Green's function.

\begin{theorem} \label{thm:convrel}
Let 
$$f(t) = \int_0^t g(s) \, \dd s.$$

Then
\begin{equation} \label{eq:convrel}
u(t,x) = \frac{1}{2 \rho} \left[ H(t - \vert x \vert/c_0) + f(t)\ast H(t - \vert x \vert/c_0)\right]
\end{equation}
and
\begin{equation}
u(t,x) = \frac{1}{2 \rho} H(t - \vert x \vert/c_0) \, [1 + \OO[t - \vert x \vert/c_0]
\end{equation}
\end{theorem}
\begin{proof}
The function $f$ is well-defined because $g \in \mathcal{L}^1_\loc([0,\infty[)$ and 
$f(t) \geq 0$, $f(0) = 0$.

It follows from eqs~\eqref{eq:Green1D} and \eqref{eq:important}  that 
$$\tilde{u}(p,x) = \frac{\kappa(p)}{2 \rho\, p^2} \e^{-p \,\vert x \vert/c_0}\, \e^{-p\, \tilde{g}(p)\, 
\vert x \vert}$$ 
But $\kappa(p) \sim_\infty p/c_0$ and
$$\frac{\kappa(p)}{p^2}-\frac{1}{p \, c_0} = \frac{\tilde{g}(p)}{p}$$
This proves eq.~\eqref{eq:convrel}.

The function $H(\cdot,r)$ is non-decreasing, hence
$$f(t)\ast H(t - r/c_0) = \int_0^{t - r/c_0} f(s)\, H(t -s - r/c_0) \, \dd s \leq 
H(t - r/c_0) \int_0^{t-r/c_0} f(s)\, \dd s$$
The integral on the right-hand side is $\OO[(t - r/c_0)]$. 
This ends the proof.

\end{proof}

For $ \vert t - r/c_0\vert \leq C_1$ the second term in the brackets is bounded from above
by a number $C_2$,
hence $\vert u(t,x)\vert \leq C \, H(t - \vert x \vert/c_0)$ for some constant 
$C > 0$. The last equation can
be used to determine upper bounds on Green's function near the wavefront. 

The function $g(t - r/c_0)\, r$ is a local phase of Green's function near its wavefront.

\section{Upper bound for Green's function near a wavefront in media with 
unbounded $g(t)$.} 

We recall that Green's functions of media with unbounded $g(t)$ do not have discontinuities.

\begin{lemma} \label{lem:gbound}
If the function $\varphi(t)$ is non-increasing and differentiable, 
$\lim_{t\rightarrow 0} [t \, \varphi(t)] = 0$,  $\;\varphi(t) \geq A $ for
some real constant $A$,  and 
the function $-t \, \varphi^\prime(t)$ is 
non-increasing, then the inverse Laplace transform of $\exp(-p\, \tilde{\varphi}(p))/p$ is 
bounded from above by $\exp(-\varphi(t))$.   
\end{lemma}
\begin{proof}
Let $f(t)$ denote the inverse Laplace transform of $\exp\left(-p\, \tilde{\varphi}(p)\right)$. 
Then $\exp\left(-p \,\tilde{\varphi}(p)\right)/p$ is the Laplace transform of 
$F(t) := \int_0^t f(s)\, \dd s$. The limit $\lim_{t\rightarrow\infty} F(t) = \int_0^\infty f(s)\, \dd s$ is finite:
\begin{multline} \label{eq:mln}
\lim_{t\rightarrow\infty} F(t) = \lim_{p\rightarrow 0} \e^{-p \tilde{\varphi}(p)} = 
\exp\left(-\lim_{p\rightarrow 0} \left[p \tilde{\varphi}(p)\right]\right) = \exp\left(-\lim_{t\rightarrow\infty} \varphi(t)\right) \\ \leq \e^{-A} < \infty
\end{multline}

We now note the identity
\begin{equation}\label{eq:ident}
t\, f(t) = -\int_0^t s\, \varphi(s)\, f(t-s)\, \dd s
\end{equation}
Indeed, the Laplace transform of the left-hand side is 
\begin{equation}
\int_0^\infty t\, f(t) \, \e^{-p \, t} \, \dd t = -\frac{\dd \tilde{f}(p)}{\dd p} = \\
\left[p \,\frac{\dd \tilde{\varphi}(p)}{\dd p} + \tilde{\varphi}(p)\right] \tilde{f}(p)
\end{equation}
while the Laplace transform of the right-hand side is 
$$
-(t \, \varphi^\prime(t))^\sim\, \tilde{f}(p) = -\left\{\left[(t \,\varphi(t))^\prime\right]^\sim - 
\tilde{\varphi}(p)\right\} \, \tilde{f}(p) = \tilde{f}(p) \left[p\, \dd \tilde{\varphi}(p)/\dd p + \varphi(p)\right]
$$
hence the two are equal. Eq.~\eqref{eq:ident} follows from the uniqueness theorems for the 
Laplace transform. 

The assumption that $-t \, \varphi^\prime(t)$ is non-increasing implies that 
$t \, f(t) \geq - t\, \varphi(t) \int_0^t f(s)\, \dd s$, hence
$$-\varphi^\prime(t) \leq \frac{\dd}{\dd t} \ln(F(t))$$
Integration of this inequality over $[t,\infty[$ yields 
the inequality $\varphi(t) \leq  -\ln(F(t))$ on account of eq.~\eqref{eq:mln},
hence $F(t) \leq \e^{-\varphi(t)}$, q.e.d.

\end{proof}

\begin{theorem} \label{thm:gbound}
Let the function $g$ defined by eq.~\eqref{eq:KK2} be such that $-t\, g^\prime(t)$ 
is non-increasing.
For $t - \vert x\vert/c_0 < \varepsilon$ there is a constant $C$ such that 
\begin{equation} \label{eq:gbound}
u(t,x) \leq C \, \exp(-g(t-\vert x\vert/c_0)\, \vert x \vert)
\end{equation}
\end{theorem}
\begin{proof}
The function $g$ defined by eq.~\eqref{eq:KK2} is LICM, hence 
it is non-increasing and $-t\, g^\prime(t) \geq 0$. It therefore has a limit at 0 which is either 
finite or $\infty$.

By Theorem~\ref{thm:convrel} for $0 \leq t - \vert x \vert/c_0 \leq C_1$, where $C_1$ is some constant, 
$$u(t,x) \leq C\, H(t - \vert x \vert/c_0, \vert x \vert)$$
where $H(t,r)$ is the inverse Laplace transform of $\exp\left(-p\, \tilde{g}(p)\, r\right)/p$.
By Lemma~\ref{lem:gbound} $H(t - r/c_0,r) \leq \exp(-r\, g(t-r/c_0))$. 

\end{proof}

\noindent\textbf{Examples.}\\
\begin{enumerate}[(1)]
\item If $\kappa(p) - p/c_0 \equiv p \, \tilde{g}(p) = a\, p^\alpha$, where $0 < \alpha < 1$ 
and $a > 0$, then $g(t) = a \, t^{-\alpha}/\Gamma(1-\alpha)$ is LICM and the function $-t\, g^\prime(t)$ is non-increasing
Hence eq.~\eqref{eq:gbound} holds.
Consequently $u(p,x) \leq C \, \exp\left(-a \, (t - \vert x \vert/c_0)^{-\alpha}\right)$.
\item $g(t) := b \, \ln(1/(a\, t) + A)$, where $a, b > 0$ have the dimensions 1/T and 1/L, $A \geq 1$, is LICM, because
\begin{equation} \label{eq:lastident}
\ln(1/(a\, t) + A) = \ln(A) + \int_0^\infty \frac{1 - \e^{-r}}{r} \e^{-r \,t/(a\, A)}\,\dd r
\end{equation}
and the second term is the Laplace transform of a positive function satisfying eq.~\eqref{eq:doss}.
The identity ~\eqref{eq:lastident} can be checked by differentiating both sides with respect to $t$ and working out the 
resulting integral on the right-hand side. Consequently $\kappa(p) = p/c_0 + p\, \tilde{g}(p)$, where $c_0 > 0$, is a complete 
Bernstein function.
The function $-t g^\prime(t)$ is non-increasing. Hence eq.~\eqref{eq:gbound} holds 
and 
$$H(t,r) \leq (1/(a\,t) + A)^{-b\,r}$$
Note that the wavefront singularity decreases stepwise by the formula  
$H(t - r/c_0) \sim_{t \rightarrow (r/c_0)+} C_1 \, (t - r/c_0)^{r\,b}$.
\item The LICM functions $(t + 1)^{-1}$ and $\e^{-t}$ do not satisfy the hypotheses of
Theorem~\ref{thm:gbound}.
\end{enumerate}

\begin{remark}
If $\kappa(p) - p/c_0$ is regularly varying with index $\alpha \in ]0,1[$ at infinity, i.e.
$\kappa(p) - p/c_0 \equiv p \, \tilde{g}(p) \sim_\infty p^\alpha\, l(p)$, where the function $l$ is slowly
varying at infinity then, by the Karamata Tauberian theorem (Appendix) $g(t) \sim_0 l(1/t)\, 
t^{-\alpha}/\Gamma(1 - \alpha)$. $g$ is non-increasing but the hypothesis that $-t \, g^\prime(t)$ is non-increasing need not
be satisfied. The function $g(t) = B\, \ln(1/(a t) + A)\, t^{-\alpha}/\Gamma(1-\alpha)$ is however LICM (because it is a product of 
two CM functions) and satisfies the conditions of Theorem~\ref{thm:gbound}. Consequently
$$\vert u(t,x) \vert \leq \left[1/[a\, (t - \vert x \vert/c_0)] + A\right]^{-B\,\vert x \vert\, (t - \vert x \vert/c_0)^{-\alpha}}$$
\end{remark} 

\section{Relations between the function $g$ and the creep rate function.}
\label{ssec:discontinuities}

 We shall show that the value of $g(0+)$ or the singularity of the function $g$ at 0 are
to some extent determined by the value of the creep rate function at 0 or its singularity at 0.
We shall thus relate the wavefront jump discontinuity or regularity of Green's function at
the wavefront to the asymptotics of creep rate at $t = 0$.

If $g(0+) < \infty$ then Corollary~\ref{cor:zcv} implies that 
$\lim_{t\rightarrow 0+} H(t,r) = \lim_{p\rightarrow \infty} \e^{-p\, \tilde{g}(p) \, r} = \exp(-g(0+)\, r)$. Consequently 
$\lim_{t\rightarrow (r/c_0))+} H(t - r/c_0,r) = \exp(-g(0+)\, r)$ and 
$\lim_{t\rightarrow (r/c_0))-} H(t - r/c_0,r) = 0$. Thus the wavefront carries a discontinuity of 
Green's function. The wavefront discontinuity can be expressed in terms of $J^\prime(0+)$.

\begin{theorem}
If $J^\prime$ has a finite limit at 0, then the function $g$ also has a finite limit at 0 and
\begin{equation}\label{eq:kli}
g(0+) = \rho \, c_0 \, J^\prime(0+)/2.
\end{equation}
Furthermore
\begin{equation} \label{eq:ineqg}
g(t) \leq \rho \, c_0 \, J^\prime(t)/2
\end{equation}
\end{theorem}
\begin{proof}
Since $\kappa(p) = p/c_0 + p \, \tilde{g}(p)$ and $\kappa(p) = \rho^{1/2} p \left[p \, \tilde{J}(p)\right]^{1/2} 
= \rho^{1/2} p \left[J_0 + \widetilde{J^\prime}(p)\right]^{1/2}$, 
using the formula $c_0 = 1/(\rho \, J_0)^{1/2}$, we obtain the identity 
\begin{equation}\label{eq:gjprime}
\frac{2}{c_0} p\, \tilde{g}(p) + p\, \left[ \tilde{g}(p)\right]^2 = \rho\, p \, \widetilde{J^\prime}(p)
\end{equation}
Since $J^\prime(0+) = \lim_{p\rightarrow\infty} \left[ p \, \widetilde{J^\prime}(p)\right]$ is finite
and $p \, \left[ \tilde{g}(p)\right]^2  \geq 0$,  the function $p\, \tilde{g}(p)$ is 
bounded from above. It follows from the proof of Lemma~\ref{lem:gpg1} that this function 
is non-decreasing, hence it tends to a limit $A$ as $p \rightarrow \infty$. Hence the 
second term on the left-hand side of equation~\eqref{eq:gjprime} tends asymptotically to 
$A^2/p$ for large $p$ and therefore it tends to zero as $p$ tends to infinity. Consequently 
$A = \rho\, c_0\, J^\prime(0)/2$, which proves the first part of the thesis.

For the second part, we note that $\tilde{g}(p)$ is the Laplace transform of a non-negative function, hence
Theorem~\ref{thm:BernsteinNew} implies that $p\,\tilde{g}(p)$ is a CM function. The second term on the left-hand
side of equation~\eqref{eq:gjprime} is the product of three CM functions ($p^{-1}$ and twice $p \, \tilde{G}(p)$),
hence it is CM. Equation~\eqref{eq:kli} implies that that 
$$\frac{\rho\, c_0}{2} p \, \widetilde{J^\prime}(p) - p\, \tilde{g}(p) \equiv \int_0^\infty \frac{\dd}{\dd t} \left[ \frac{\rho\, c_0}{2} J^\prime(t) -
g(t) \right] \, \e^{-p\, t} \, \dd t. $$
By Bernstein's theorem \cite{BernsteinFunctions} the above expression is the Laplace 
transform of a positive Radon measure $\mu$. By the uniqueness of the Laplace transform 
$$\mu(\dd t) = \frac{\dd}{\dd t} \left[ \frac{\rho\, c_0}{2} J^\prime(t) -g(t)\right] \, \dd t$$
and therefore $$\frac{\dd}{\dd t} \left[ \frac{\rho\, c_0}{2} J^\prime(t) -g(t)\right] \, \dd t \geq 0$$
The last inequality and equation~\eqref{eq:kli} imply inequality~\eqref{eq:ineqg}, q.e.d.

\end{proof}

Eq.~\eqref{eq:kli} can also be expressed in the form
\begin{equation} \label{eq:g0vsJ}
g(0+) = J^\prime(0+)/(2 J_0\, c_0)
\end{equation}

It follows from equation~\eqref{eq:duality} that $J_0\, G_0 = 1$ and $J^\prime(0+) \, G_0 + G^\prime(0+) \, J_0 = 0$,
where $G_0 = \lim_{t \rightarrow 0} G(t) < \infty$ in view of our assumption that $J_0 > 0$. The above identities imply 
a third expression for $g(0+)$:
\begin{equation} \label{eq:g0vsJ1}
g(0+) = -G^\prime(0+)/(2\, \rho\, c_0^{\;3})
\end{equation}
which is consistent with Chu's amplitude equation for a shock wave \cite{Chu62}. Under our assumptions $g(0+) = J^\prime(0+) = \infty$ if 
and only if $G^\prime(0+) = -\infty$. The latter inequality was given as a criterion for non-existence of 
shocks by Pr\"{u}ss in \cite{Pruss87}.

Equation~\eqref{eq:ineqg} implies that the singularity of $g$ at 0 is not stronger than the singularity of 
the creep rate $J^\prime$. 

Equations~\eqref{eq:convrel}, \eqref{eq:asymp} and \eqref{eq:g0vsJ} imply that the jump discontinuity 
at the wavefront is given by the expression $(2 \rho)^{-1}\, \exp(-J^\prime(0+)\, r/(2 J_0\, c_0))$,
which is a well known result \cite{Christensen}.

\section{Conclusions.}

The causal function $g(t)$ introduced in \cite{SerHanJMP,HanWM2013} in the context of Kramers-Kronig relations
has now been used to study the wavefront asymptotics of Green's functions. The function $g$ is majorized by the 
creep rate function $J^\prime$ multiplied by the factor $\rho\, c_0/2$ and its value at 0 can be expressed in 
terms of $J^\prime(0)$. Consequently discontinuities are expected at the wavefront if and only if $J^\prime(0) < \infty$.
Use of $g$ allows for more detailed estimates than it would be possible with constitutive parameters only.

Non-linearity can generate shock waves in elastic media through gradient catastrophe. 
Our analysis of wavefront singularities suggests that this process might be impeded by 
an unbounded viscoelastic attenuation. So far little has been done to clarify 
the competition between non-linearity and various kinds of viscoelastic dissipation.

\appendix
\section{Recapitulation of necessary mathematical concepts.}  

An infinitely differentiable function $g:\; ]0,\infty[ \rightarrow \mathbb{R}$ is 
{\em completely monotonic} (CM) if 
$(-1)^n \D^n g(t) \geq 0$ for $t > 0$ and $n \in \mathbb{Z}_+ \cup \{0\}$. By the Leibnitz 
formula the product $f(t)\, g(t)$ of two CM functions $f$ and $g$ is CM.

The CM function $g$ can have a singularity at 0 but it is integrable on any finite interval 
not including 0. The function $g$ is therefore locally integrable if and only if
$\int_0^1 g(t) \, \dd t < \infty$. In view of non-negativity, monotonicity and continuity 
$g$ always has a limit $g_\infty = \mu\{ 0\} \geq 0$ at $\infty$.

Bernstein's Theorem \cite{BernsteinFunctions} asserts 
that $g$ is a completely monotonic function if and only if $g$ is the Laplace transform of a positive Radon measure 
(essentially a locally finite measure) 
$\nu$ such that 
$$g(t) =  \int_{[0,\infty[} \e^{-r\, t} \nu(\dd r)$$
$g$ has a finite limit at 0 if $\nu$ has finite mass and then $\lim_{t\rightarrow 0} g(t) = g(0+) = \nu([0,\infty[)$. 
$g$ is locally integrable if and only if $\nu$ satisfies the inequality \eqref{eq:doss} \cite{GripenbergLondenStaffans}.

A non-negative function $f$ on $]0,\infty[$ is called a {\em Bernstein function} if it is differentiable and its derivative 
is completely monotonic. Monotonicity and continuity of $f$ imply that it has a finite non-negative limit at 0.

A function $h$ on $]0,\infty[$ is said to be a {\em complete Bernstein function} if it has the form $h(p) = p^2 \, \tilde{F}(p)$,
where $F$ is a Bernstein function. The function $h$ is a complete Bernstein function if and only if
it has the integral representation
$$h(p) = a + b\, p + p \int_{]0,\infty[} \frac{\nu(\dd r)}{p + r}$$
where $a$ and $b$ are non-negative constants and $\nu$ is a positive Radon measure satisfying inequality~\eqref{eq:doss}
\cite{BernsteinFunctions}.

A measurable real function $f$ on $[0,\infty[$ is said to be {\em regularly varying with index} $\gamma \in \mathbb{R}$ at $a = 0$ or $\infty$ if $\lim_{t\rightarrow a} f(\lambda\, t)/f(t) = \lambda^\gamma$ for $\lambda > 0$ \cite{BinghamGoldieTeugels}. The function $f$ is said to be {\em slowly varying} at $a$ if it is
regularly varying with index 0 at $a$. A general function $f$ regularly varying with index $\gamma$
at $a$ has the form $f(t) = t^\alpha\, l(t)$, where $l(t)$ is slowly varying at $a$.

We recall the Karamata Tauberian theorem (\cite{BinghamGoldieTeugels}, a corollary of Theorems 1.7.1 and 1.7.2):
\begin{theorem}
If $f \in \mathcal{L}^1_{\mathrm{loc}}([0,\infty[)$ is non-negative and monotone for $t > T$, where $T$ is a positive number,
$\alpha \geq 0$ and the function $l(t)$ is slowly 
varying at infinity then 
then $f(t) \sim_\infty t^{\alpha-1} \, l(t)$ is equivalent to $\tilde{f}(p) \sim_0 p^{-\alpha}\, l(1/p)$. 
\end{theorem}

\end{document}